\documentclass[12pt]{article}

\usepackage[margin=1in]{geometry}
\usepackage{amsmath,amsfonts,graphicx,framed}
\usepackage{amssymb}
\usepackage{amsthm}
\usepackage{fancyhdr}
\usepackage{float}
\usepackage{caption}
\usepackage{comment}
\usepackage[utf8]{inputenc}

\usepackage{authblk}
\usepackage[hypertexnames=false]{hyperref}

\theoremstyle{plain}
\newtheorem{thm}{Theorem}

\newtheorem{cor}{Corollary}

\newtheorem{prop}{Proposition}

\theoremstyle{definition}
\newtheorem{defn}{Definition}

\theoremstyle{remark}
\newtheorem{rem}{Remark}

\numberwithin{equation}{section}
\numberwithin{prop}{section}

\newcommand{\X}{\,\text{X}}
\newcommand{\Xe}{\,\emph{X}}
\newcommand{\pd}{\,\partial}

\newcommand{\al}{\,\alpha}

\newcommand{\eqal}[1]{\begin{equation}\begin{aligned}#1\end{aligned}\end{equation}}

\begin{document}

\title{Sub-symmetries II. Sub-symmetries and Conservation Laws.}

\author{V. Rosenhaus*}                    
\author{Ravi Shankar}

\affil{\footnotesize{* Department of Mathematics and Statistics, California State University, Chico, CA, USA, vrosenhaus@csuchico.edu}}
\affil{\footnotesize{Department of Mathematics, University of Washington, Seattle, WA, USA, 
shankarr@uw.edu}}
\date{}

\maketitle

\begin{abstract}
In our previous paper, the concept of sub-symmetry of a differential system was introduced, and its properties and some applications were studied. It was shown that sub-symmetries are important in decoupling a differential system, and in the deformation of a system's conservation laws, to a greater extent than regular symmetries. 

In this paper, we study the nature of a correspondence between sub-symmetries and conservation laws of a differential system. We show that for a large class of non-Lagrangian systems, there is a natural association between sub-symmetries and local conservation laws based on the Noether operator identity, and we prove an analogue of the first Noether Theorem for sub-symmetries. We also demonstrate that infinite conservation laws containing arbitrary functions of dependent variables can be generated by infinite sub-symmetries through the Noether operator identity.  

We discuss the application of sub-symmetries to the incompressible 
Euler equations of fluid dynamics. Despite the fact that infinite symmetries (with arbitrary functions of dependent variables) are not known for the Euler equations, we show that these equations possess infinite sub-symmetries. We calculate infinite sub-symmetries with arbitrary functions of dependent variables for the two- and three-dimensional Euler equations in the velocity and vorticity formulations with certain constraints. We demonstrate that these sub-symmetries generate known series of infinite conservation laws, and obtain new classes of infinite conservation laws.

\end{abstract}

\section{Introduction}

\smallskip

A concept of sub-symmetry was introduced in \cite{RS2016}. A sub-symmetry is a symmetry of a part of the system (its sub-system) on solutions of the entire system. In was shown in \cite{RS2016} that sub-symmetries can allow a decoupling of a sub-system from the rest of the system. Another attractive feature of sub-symmetries compared to regular symmetries is the ability of sub-symmetry transformations to generate larger sets of conservation laws as deformation of a given conservation law. For example, for the nonlinear telegraph system it was demonstrated in \cite{RS2016} that sub-symmetries can generate all lower conservation laws through the deformation of a known conservation law.

In this work, we study a mechanism of association between sub-symmetries and local conservation laws. Our approach is based on the Noether operator identity \cite{Rosen} that relates 
the infinitesimal transformation operator to the Euler and total derivative operators.  The Noether operator identity provides a Noether-type relation between symmetries and conservation laws not only for Lagrangian systems, see e.g. \cite{Olver}, but also for a large class of differential systems that are not known to have a well-defined variational functional, see \cite{Rosenhaus94}, \cite{Rosenhaus96}. In this paper, we extend this approach 
to sub-symmetries and show that the Noether operator identity provides a natural association between sub-symmetries of a differential system and its conservation laws. A similar association can be established for infinite sub-symmetries of the system. The main focus of this paper is the relation between infinite sub-symmetries and infinite conservation laws of a differential system, which we demonstrate 
for the Euler equations.
  
Many conservation laws for incompressible inviscid fluid flows described by the Euler equations are known, including some infinite-dimensional sets of conservation laws (infinite conservation laws) parametrized by arbitrary functions of dependent variables (velocities or vorticities). Among them are recently obtained classes of infinite conservation laws for the helically invariant Euler equations \cite{Kelbin}. 

Using a Hamiltonian formulation, some
known conservation laws have been associated with the symmetries of the Euler equations (linear momentum, angular momentum, and energy; see \cite{Olver82}).  
At the same time, it was shown that the infinite ``Casimir" conservation laws in 2D and the helicity conservation law in 3D can be explained, not by symmetries, but by a degeneracy of the Poisson bracket, \cite{Olver82}.  The nature of these conservation laws and their connection to the symmetry properties of the original system remained, essentially, unknown. The same conclusion is true for a number of recently obtained infinite series of conservation laws \cite{Kelbin}. To a large extent, it is, apparently, due to the absence of a well-defined Lagrangian function for the Euler equations.

\smallskip

In this paper, we study infinite-dimensional symmetries (sub-symmetries) and conservation laws of the Euler equations.  By infinite-dimensional (infinite), we mean those symmetries (sub-symmetries) and conservation laws that contain arbitrary functions of the dependent variables. 
We show that the existence of certain sub-symmetries of the Euler system provides an explanation of many known conservation laws of the Euler equations, including all infinite conservation laws with arbitrary functions of dependent variables.

\smallskip

Infinite symmetry algebras parametrized by arbitrary functions and their relationships to conservation laws have been described in the literature considerably less extensively than finite-dimensional Lie symmetry groups and their corresponding conservation laws. Most known results concern differential systems arising from some variational problem with a Lagrangian function. According to the Second Noether Theorem \cite{Noether} (see also \cite{Olver}), infinite variational symmetries with arbitrary functions of {\it all} independent variables lead to certain identity relationships between the equations of the original differential system and their derivatives (meaning that the original system is under-determined). Infinite variational symmetries with arbitrary functions of \textit{not all} independent variables were studied in \cite{Rosenhaus02} and were shown to lead to a finite number of essential (integral) local conservation laws. Each essential conservation law is determined by a specific type of boundary condition. Boundary conditions and finite sets of essential conservation laws were found for a number of physically interesting differential systems, e.g. \cite{Rosenhaus06a, Rosenhaus06b, Rosenhaus03a}. 

The situation with infinite symmetry algebras parametrized by arbitrary functions of \textit{dependent variables} was shown in \cite{Rosenhaus07}
to be radically different;  in this case, an infinite symmetry leads to an infinite set of essential conservation laws.
Known examples of this situation are linear equations,
equations of the Liouville type, see e.g \cite{Zhiber1979a, Zhiber2001} that can be integrated by the Darboux method (e.g. \cite{Sokolov1995, Juras1997}) and equations of the hydrodynamic type \cite{Sheftel2004,Grundland2000} that can be integrated by the generalized hodograph method. The general case of one scalar variational equation of the second order admitting infinite symmetries and infinite conservation laws with arbitrary functions of a dependent variable $u(t,x)$ and its first derivatives $u_t, u_x$ was analyzed in \cite{Rosenhaus07}, and an extension of this work to the case of systems of two variational equations was discussed in \cite{Rosenhaus09} (see \cite{Rosenhaus13} for a general approach and a connection to equations of the Liouville type). 

In the recent paper \cite{Rosenhaus15}, we studied 
differential systems that may not have well-defined Lagrangian functions (\emph{quasi-Noether} systems), but possess infinite symmetries involving arbitrary functions of \textit{all} independent variables.  It was shown that the approach based on the Noether operator identity allows for the derivation of an extension of the Second Noether theorem for non-Lagrangian differential systems, \cite{Rosenhaus15}.

\smallskip
\medskip

In this paper, we analyze infinite sub-symmetries and a corresponding infinite set of local conservation laws of a non-variational system. In Section \ref{sec:var}, we review the well-known correspondence between symmetries and conservation laws for variational systems from the standpoint of the Noether identity.  
We also discuss this correspondence for a more general class of differential systems (quasi-Noether) that may not arise from variational problems. We apply the approach based on the Noether identity and show how infinite symmetries of quasi-Noether systems give rise to infinite sets of conservation laws.  In Section \ref{sec:sub}, we formulate and prove a theorem relating sub-symmetries and conservation laws of quasi-Noether systems.  In Sections \ref{sec:2d}-\ref{sec:vort}, we apply this approach to the systems of Euler and constrained Euler equations and find infinite sub-symmetries and infinite conservation laws of these systems.  
In Section \ref{sec:2d}, we study an infinite sub-symmetry for the two-dimensional vorticity equations, and show how this infinite sub-symmetry generates the well-known Casimir invariants of 2D flow. In Section \ref{sec:euler}, we
study three-dimensional Euler systems with constraints and find general classes of infinite sub-symmetries and infinite conservation laws containing arbitrary functions of velocities. 
We generate new classes of infinite conservation laws  
and show that our results include recently found series of infinite conservation laws of \cite{Kelbin}.  In Section \ref{sec:vort}, we apply this methodology to the 3D Euler equations in vorticity form, construct 3D constrained vorticity systems admitting infinite sub-symmetries and conservation laws with arbitrary functions of vorticity components, obtain
the results reported in \cite{Kelbin}, and generate new classes of infinite conservation laws.

\section{Symmetries, conservation laws, and the Noether operator identity}
\label{sec:var}

Let us briefly outline the approach we follow; for details see \cite{Rosenhaus02,Rosenhaus07}.  By a conservation law of a differential system 
\begin{align}\label{sys}								
\Delta_v(x,u,u_{(1)},u_{(2)},\dots, u_{(l)})=0, \quad v=1,2,...,n.
\end{align}
we mean a divergence expression
\begin{align}\label{conser}
D_i K^i(x,u,u_{(1)},u_{(2)},\dots ) \doteq 0,
\end{align}
that vanishes on all solutions of the original system ($\doteq$).  Here, $x=(x^1,x^2,\dots,x^p)$ and $u=(u^1,u^2,\dots,u^q)$ are the tuples of 
independent and dependent variables, respectively; 
$u_{(r)}$ is the tuple of $r$th-order derivatives of $u, \:\:\:r=1,2,\dots$;
$\Delta_v$ and $K^i$ are differential functions, i.e. smooth functions of $x^i$, $u^a$ and a finite number of derivatives of $u$ (see \cite{Olver}); $i = 1,\dots, p$, $a=1,\dots,q$. $D_i$ is the total derivative with $x^i$:
\begin{align*}
D_i=\partial_i+u^a_i\partial_{u^a}+u^a_{ij}\partial_{u^a_{ij}}+\dots=\partial_i+u^a_{iJ}\partial_{u^a_J},\quad 1\le i\le p, \quad 1\le a\le q, 
\end{align*}
the sum extending over all (unordered) multi-indices $J=(j_1,j_2,...,j_k)$ for $k\ge 0$ and $1\le j_k\le p$, as well as over $1\le a\le q$. (For $k=0$ we set $u^a_J=u^a$ and $\mathrm D_J=1$).
\smallskip

Two conservation laws $K$ and $\tilde K$ are equivalent if they differ by a trivial conservation law \cite{Olver}.  A conservation law $D_i P^i \doteq 0 $ is trivial if a linear combination of two kinds of triviality is taking place: 1. The $p$-tuple $P $ vanishes on the solutions of the original system: $P^i \doteq 0$. 2. The divergence identity holds for any smooth 
$q$-tuple of functions $ u $ (e.g. $ \partial_t (\text{div} u) + \text{div}(-\partial_t u) = 0$ for $p=q$ and $u=u(t,x)$).
\medskip

Let smooth functions $u^a=u^a(x)$ be defined on a non-empty open subset $D$ of $p$-dimensional space.  
\smallskip

We start with the case when the system \eqref{sys} is of a variational problem with the action functional
\[
S ={\int_D} {L(x,u, u_{(1)}, \dots)\: d^{p}x}, 
\]											
where $L$ is 
the Lagrangian density. Then
\begin{equation}\label{Rosenhaus:equation2}	
E_a(L)\equiv \Delta_a(x,u, u_{(1)}, \dots) = 0,\quad 1\le a\le q,  
\quad (q=n),
\end{equation}										
where 
\begin{equation}\label{Rosenhaus:equation3}	
E_a 
= \frac{\partial }{\partial u^a} - \sum\limits_i {D_i \frac{\partial }{\partial u_i^a } }
+ \sum\limits_{i \leqslant j} {D_i } D_j \frac{\partial }{\partial u_{ij}^a } + \cdots 
=(-D)_J\frac{\partial}{\partial_{u^a_J}},
\end{equation}
is the $a$-th Euler--Lagrange operator (variational derivative), and $(-\mathrm D)_J$ is the adjoint operator: $(-D)_J=(-1)^kD_J=(-D_{j_1})(-D_{j_2})\cdot\cdot\cdot(-D_{j_k})$. 
The operator $E_a$, clearly, annihilates total divergences.  

Consider an infinitesimal transformation with the canonical infinitesimal operator
\begin{align}\label{Rosenhaus:equation4}	
X_\alpha = \alpha ^a &\frac{\partial }{\partial u^a} + \sum\limits_i {(D_i 
\alpha ^a)\frac{\partial }{\partial u_i^a }} + \sum\limits_{i \leqslant j} {(D_i } 
D_j \alpha ^a)\frac{\partial }{\partial u_{ij}^a } + \cdots\, = \, \mathrm (D_{J}{\alpha}^a) \partial_{u^a_J},
\end{align}
where 
$\alpha ^a = \alpha ^a (x,u, u_{(1)}, \dots$), and the sum is taken over all (unordered) multi-indices $J.$
The variation of the functional $S$ under the transformation with operator 
$X_\alpha $ is
\begin{equation}								\label{Rosenhaus:equation5}	
\delta S = \int_D{X_\alpha L  \,d^{p}x}\,.
\end{equation}										
If $X_\alpha $ is a variational (one-parameter) symmetry transformation then 
\begin{equation}								\label{Rosenhaus:equation6}	
X_\alpha L =  D_i M^i,
\end{equation}										
where $M^i=M^i (x,u,u_{(1)},\dots)$ are smooth functions of their arguments. 
The Noether operator identity \cite{Rosen} (see also, e.g. \cite{Ibragimov1985} 
or \cite{Rosenhaus94}) relates the operator $X_\alpha$ to $E_a$,
\begin{equation} 								\label{Rosenhaus:equation7}	
X_\alpha = \alpha^aE_a +{D_i R^{i} }, 		
\end{equation}										
\begin{equation}	\label{Rosenhaus:equation8}	
R^{i} = \alpha ^a\frac{\partial }{\partial u_i^a } + \left\{ 
{\sum\limits_{k \geqslant i} {\left( {D_k \alpha ^a} \right) - \alpha 
^a\sum\limits_{k \leqslant i} {D_k } } } \right\}\frac{\partial }{\partial 
u_{ik}^a }+ \cdots.									
\end{equation}
The expression for $R^{i}$ can be presented in a more general form \cite{Rosen,Lunev}:
\begin{align}\label{Ri}
R^i=(D_K\al^a)\,(-D)_J\pd_{u^a_{iJK}},
\end{align}
where $J$ and $K$ sum over multi-indices.

Applying the identity \eqref{Rosenhaus:equation7} with 
\eqref{Rosenhaus:equation8} to $L$ and using \eqref{Rosenhaus:equation6},	  
we obtain
\begin{equation}								\label{Rosenhaus:equation9}	
D_i (M^i - R^i L) = \alpha ^a\Delta_a, 
\end{equation}										
which on the solution manifold ($\Delta = 0$, $D_i \Delta = 0$, \dots ) 
\begin{equation}								\label{Rosenhaus:equation10}	
D_i (M^i - R^i L)\doteq 0,   
\end{equation}										
leads to the statement of the First Noether Theorem: any 
one-parameter variational symmetry transformation with infinitesimal 
operator $X_\alpha$ \eqref{Rosenhaus:equation4} gives 
rise to
the 
conservation law \eqref{Rosenhaus:equation10}.
\medskip

Note that Noether \cite{Noether} used the identity \eqref{Rosenhaus:equation9} and not the operator identity \eqref{Rosenhaus:equation7}.  The first mention of the Noether operator identity  \eqref{Rosenhaus:equation7}, to our knowledge, was made in \cite{Rosen}.
 
\smallskip

Consider now differential systems without well-defined Lagrangian  functions.

\smallskip

In \cite{Rosenhaus94} and \cite{Rosenhaus96}, 
an approach based on the Noether operator identity \eqref{Rosenhaus:equation7} was suggested to relate symmetries to
conservation laws for a large class of differential systems that may not 
have well-defined Lagrangian functions; in \cite{Rosenhaus15} these systems were called quasi-Noether. 

\smallskip

In the current paper, we discuss the application of the approach based on the Noether operator identity to sub-symmetries, and the 
association between sub-symmetries and conservation laws for quasi-Noether systems.

\smallskip

Let us start with the symmetries of system \eqref{sys}. Applying the Noether operator identity \eqref{Rosenhaus:equation7} 
to a linear combination of equations $\Delta_a$ with coefficients $\beta^a$, we obtain 
\begin{equation}								\label{Rosenhaus:equation12}	
X_\alpha {(\beta^v\Delta_v)}= \alpha^a E_a (\beta^v\Delta_v) +{D_{i} R^i(\beta^v\Delta_v)},
\qquad a =1, \dots q, \quad v=1, \dots, n, \quad i=1, \dots, p.
\end{equation}
If there exist coefficients 
$\beta^a$ such that (\cite{Rosenhaus96})
\begin{equation}\label{quasi-Noether}	
E_a (\beta^v\Delta_v) \doteq 0, \qquad a = 1, \dots q, \quad v=1, \dots, n,
\end{equation}
then each symmetry $X_\alpha$ will lead to a local conservation law 
\begin{equation}								\label{Rosenhaus:equation16}	
D_i R^i (\beta^v \Delta_v)\doteq 0, 
\end{equation}
for any differential system of class 
\eqref{quasi-Noether}. Thus, the condition \eqref{quasi-Noether} can be considered as defining quasi-Noether systems, see also \cite{Rosenhaus15}. 
A system \eqref{euler:sys} is \emph{quasi-Noether} if there exist functions (differential operators) $\beta^a$ such that the condition \eqref{quasi-Noether} is satisfied. In \cite{Rosenhaus94}, 
the quantity $\beta^v\Delta_v$ was referred to as an alternative Lagrangian.
 
Quasi-Noether systems are differential systems for which it is possible 
to associate local conservation laws to symmetries based on the Noether operator identity, see \cite{Rosenhaus15}. Quasi-Noether differential systems are rather general and include all systems possessing conservation laws. Quasi-Noether systems include all differential systems in the form of conservation laws, e.g. KdV, mKdV, Boussinesq, Kadomtsev-Petviashvili equations, nonlinear wave and heat equations, Euler equations, and Navier-Stokes equations; as well as the homogeneous Monge-Ampere equation, and its multi-dimensional analogue, see \cite{Rosenhaus94}. In \cite{Rosenhaus15}, the approach based on the Noether identity was applied to quasi-Noether systems possessing infinite symmetries involving arbitrary functions of all independent variables, in order to generate an extension of the Second Noether theorem for systems that may not have well-defined Lagrangian functions.
\smallskip

\medskip

Note that the condition \eqref{quasi-Noether} (for a system to be quasi-Noether and allow a Noether-type association between symmetries and conservation laws) was obtained earlier within an alternative approach based on the Lagrange identity, see \cite{Vladimirov1980}, \cite{Vinogradov}, \cite{Zharinov1986}, \cite{Caviglia1986}, \cite{Sarlet1987}, \cite{Lunev}.  

\smallskip
Note also that the condition used by the authors of the direct method \cite{AncoBluman1997} (which aims at the generation of conservation laws for a differential system without regard to its symmetries), and that of the nonlinear self-adjointness approach \cite{Ibragimov2011}
is equivalent to the quasi-Noether condition \eqref{quasi-Noether}.

\section{Sub-symmetries and conservation laws}
\label{sec:sub}

Consider an infinitesimal transformation with operator $\text{X}$ 
\begin{equation}		\label{non-canon_X}
\text{X}=\xi^i\partial_i+\phi^a\partial_{u^a}+ \ldots, 
\end{equation}
where 
$\text{X}$ 
is a correspondingly prolonged vector field, see e.g. \cite{Olver}.

In \cite{RS2016} a \textit{sub-symmetry} of the system was introduced as an infinitesimal transformation of a nonempty subset (sub-system) of the system that leaves the subset invariant on the solution set of the entire system. 

Let us illustrate our motivation to study conservation laws in connection with sub-symmetries with the following observation. For a system of PDE's \eqref{sys} any conservation law is determined by some differential combination of the original equations:
\begin{align}\label{combin}
D_iM_i = \Gamma_v \Delta_v, 
\end{align}
where $\Gamma_v ={\Gamma^J}_v D_J$ are differential operators. In a regular approach, conservation laws of 
a differential system are associated with the symmetries of this system. However, it seems natural to study the dependence 
of conservation laws \eqref{combin} on the invariance properties of the specific combination of the equations of the system $\Gamma_v \Delta_v$ rather than those of the whole system $\Delta_v = 0, \:\: v=1,...,n$.  Consider for example, the nonlinear Schr\"{o}dinger  (NLS) equations 
\begin{align} \label{NLSC}
i\psi_{t} + \psi_{xx}  -k \psi^2 \psi^{*}  =  0, \nonumber \\
-i\psi^{*}_{t} + \psi^{*}_{xx}  -k \psi {\psi^{*}}^2  =  0,
\end{align}

\noindent \negthinspace \negthinspace where $\psi=\psi(x,t)=u+iv$,  $\psi^{*}=\psi^{*}(x,t)=u-iv$,  
and $u=u(x,t)$, and $v=v(x,t)$ are
real-valued functions, and $k$ is a  real constant.   
\noindent The nonlinear Schr\"{o}dinger equations for functions $u$ and $v$ take the form 
\begin{align}\label{NLSR}
\Delta_1\:&=\:-v_{t} + u_{xx}  -k u (u^2 +v^2)  =  0, \nonumber \\
\Delta_2\:&=\: \: \: \,u_{t} + v_{xx}  -k v (u^2 +v^2) =  0.
\end{align}
\noindent For the NLS system \eqref{NLSR}, the following 
combination of original equations: $ -v\Delta_1 + u\Delta_2 $ is known to 
lead to a continuity equation (conservation law) on solutions ($ \: \doteq$)
\begin{equation}                                \label{6.13}
D_t (u^2 +v^2)+2D_x\left(uv_{x}-vu_{x}\right) \:  =\: -v\Delta_1 + u\Delta_2 \doteq \; 0.
\end{equation}
The corresponding integral conservation law (associated with appropriate
boundary conditions, see \cite{Rosenhaus16}) is a conservation of the number of particles (mass).
\smallskip
\begin{equation} \label{number_particles}
D_t \!\int {(u^2 +v^2) \: dx  } \doteq 0.
\end{equation}
The dilatation
\eqal{
\X=u\pd_u+v\pd_u,
}
is not a symmetry of system \eqref{NLSR}.  However, this operator is a symmetry of the sub-system $-v\Delta_1+u\Delta_2=0$, and it deforms equation combination \eqref{6.13} into a conservation law (itself)
\eqal{
\X(-v\Delta_1+u\Delta_2)=2(-v\Delta_1+u\Delta_2)\doteq 0.
}
Thus, in this case, the operator $\X$ that generated the conservation law is a \textit{sub-symmetry} of the system \eqref{NLSR}.  

For variational systems, all local conservation laws arise from symmetry deformations of the Lagrangian (variational symmetries, see \cite{Olver}).  For systems without a Lagrangian, conservation laws can instead arise from sub-symmetry deformations of some combinations of original equations.  As introduced in \cite{RS2016}, a symmetry of a combination of equations of the system (or sub-set of the system) is a sub-symmetry of the system.

\begin{defn}[Sub-symmetry]
A transformation with infinitesimal operator $\text{X}$ \label{X} is a \textit{sub-symmetry} of the system \eqref{sys} if, for some given non-zero linear differential operators $\Xi^{ib}$, there exist linear differential operators $\Lambda^{ib}$ with smooth coefficients such that following relationship holds for any functions $u$:
\begin{align}\label{sub_def}
\text{X}\left(\Xi^{iv}\Delta_{v}\right)=\Lambda^{iv}\Delta_v,\hspace{4 mm}i=1,2,...,m, \quad v=1,\ldots,n, \qquad 0<m \le n,
\end{align}
where $\Xi^{iv}=\Xi^{ivJ}[u]D_J$, \:\: $\Lambda^{iv}=\Lambda^{ivJ}[u]D_J, \quad D_J=D_{j1}D_{j2}\cdot\cdot\cdot D_{jk}$, and each $ D_{jr},  \:r=1,2,. . ., k$ is a total derivative operator. 
\end{defn}
Since a sub-symmetry $X$ of the system \eqref{sys} is defined for some particular sub-system $(\Xi^{iv}\Delta_{v})$, we say that the pair $(X,\, \Xi^{iv}\Delta_{v})$ is a sub-symmetry of the system \eqref{sys}.
\smallskip

Consider a canonical (evolutionary) operator corresponding to the operator \eqref{non-canon_X} 
\begin{align}
\hat{\text{X}}=\text{X}-\xi^iD_i=(\phi^a-u^a_i\xi^i)\partial_{u^a}\equiv\hat{\phi}^a\partial_{u^a}.
\end{align}
We will use a regular notation $X$ for a canonical operator. Its prolongation is
\begin{align}											\label{sub-symm}
\text{X}=D_J{\phi}^a\partial_{u^a_J},
\end{align}
where the sum is taken over all (unordered) multi-indices $J$.

The following theorem provides a correspondence between sub-symmetries of quasi-Noether 
differential systems \eqref{sys} and their local conservation laws. According to the
condition \eqref{quasi-Noether} for a system to be quasi-Noether, there exist differential functions $\Xi^v[u]$ such that on the solutions ($\Delta=0$)
\begin{align}\label{quasi-N}
E_a\left(\Xi^v\Delta_v\right) \doteq 0, \qquad 1\le a\le q,
\end{align}
(where the sum is taken over $1\le v \le n$). The condition \eqref{quasi-N} means that there exist differential functions $\Xi^v[u]$, and differential operators $\Gamma^{av}=\Gamma^{avJ}D_J$ 
such that the following relation holds: $E_a\left(\Xi^v\Delta_v\right)=\Gamma^{av}\Delta_v$  
for  $1\le a\le q$, where $p$ and $q$ are the number of independent, and dependent variables, respectively, and $n$ is the number of
equations in the system \eqref{sys}. 

\begin{thm}[Sub-symmetries and conservation laws]						\label{thm:subcon}
									
Let $\Delta$ be a quasi-Noether differential system \eqref{sys}: 
\begin{align}\label{euler_id}
&\hspace{10 mm}E_a\left(\Xi^v\Delta_v\right)=\Gamma^{av}\Delta_v,\hspace{10 mm}1\le a\le q.
\end{align}  
For any sub-symmetry $(X,\,\Xi^v\Delta_{v})$ of the system $\Delta$ 
\begin{align}\label{sym_id}
&\hspace{10 mm}{\text{\emph{X}}}\left(\Xi^v\Delta_v\right)=\Lambda^v\Delta_v,   \qquad \Lambda^{v}=\Lambda^{vJ}D_J,
\end{align}
there is a corresponding conservation law 
\begin{align}\label{cons_law}
&\hspace{10 mm}D_iR^i\left(\Xi^v\Delta_v\right)\doteq 0,
\end{align}
where $R^i$ is given by \emph{(\ref{Ri})}.
\end{thm}

\begin{proof}
By the Noether identity \eqref{Rosenhaus:equation7}, condition (\ref{sym_id}) can be 
rewritten as follows:
\begin{align}\label{proof1}
{\phi}^aE_a\left(\Xi^v\Delta_v\right)+D_iR_i\left(\Xi^v\Delta_v\right)=\Lambda^v\Delta_v.
\end{align}
Substituting condition (\ref{euler_id}) into \eqref{proof1} 
\begin{align*}
{\phi}^a\Gamma^{av}\Delta_v+D_iR_i\left(\Xi^v\Delta_v\right)=\Lambda^v\Delta_v,
\end{align*}
we obtain the following continuity equation:
\begin{align}\label{con}
D_iR^i\left(\Xi^v\Delta_v\right)=\left(\Lambda^v-{\phi}^a\Gamma^{av}\right)\Delta_v=0
\end{align}
when $\Delta=0$.
\end{proof}

\begin{rem}
This theorem is an analogue of the first Noether theorem for sub-symmetries. However, while nontrivial variational symmetries (whose coefficients $\phi^a$ do not vanish on solutions) always give rise to nontrivial conservation laws (the characteristics of variational symmetries are the same as the characteristics of the corresponding conservation laws, see \eqref{Rosenhaus:equation9}, see also \cite{Olver}), nontrivial sub-symmetries may lead to trivial conservation laws. Consider the differential operators from the RHS of \eqref{con}
\begin{equation}
\Lambda^v-{\phi}^a\Gamma^{av}\equiv \mu^v=\mu^{vJ}[u]D_J,
\end{equation} 
Then 
\begin{align}													
\mu^v \Delta_v = \mu^{vJ}D_{J}\Delta_v = 
D_{J}\left(\mu^{vJ} \Delta_v \right) + \left((-\mathrm D)_{J}\,\mu^{vJ}\right) \Delta_v.	
\end{align}
Since the first term in the RHS is a trivial conservation law of the first type, conservation law \eqref{con} takes the form:
\begin{align}
D_iR^i\left(\Xi^v\Delta_v\right)=\bar\mu^v\Delta_v
\end{align}
(up to trivial conservation laws), where each $\bar\mu^v$ is a differential function. Thus, a nontrivial sub-symmetry \eqref{sub-symm} leads to a trivial conservation law if
\begin{align}\label{nontriv:sub}
\bar\mu^v =\left((-\mathrm D)_{J}\mu^{vJ}\right)=0, \qquad v=1,...,n 
\end{align}
when $\Delta=0$.  Respectively, a nontrivial sub-symmetry \eqref{sub-symm} generates a nontrivial conservation 
law \eqref{con} if
\begin{align}
\bar\mu^v=(-D)_J\mu^{vJ}\ne 0, \qquad v=1,...,n 
\end{align}
when $\Delta=0$.  
\end{rem}

\begin{rem}
Condition (\ref{euler_id}) is always satisfied for conservation systems: $\Delta_v=D_i M_{iv}$ 
\cite{Rosenhaus94}.  Condition (\ref{sym_id}) requires that the transformation $\hat{\text{X}}$ 
be a sub-symmetry of the combination of the equations of the system $\Xi^v\Delta_v$ with 
respect to the whole system $\Delta$. This condition is always satisfied if $\hat{\text{X}}$ is 
a symmetry of $\Delta$.
\end{rem}

\begin{rem}
In addition to finite conservation laws corresponding to $n$-parameter sub-symmetries, the theorem also 
allows for the generation of infinite conservation laws corresponding to sub-symmetries with arbitrary functions of
dependent variables. In this case, we obtain infinite conservation laws with arbitrary functions of dependent variables. This situation is different from the case of the second Noether theorem which involves 
arbitrary functions of all independent variables, 
see \cite{Rosenhaus15}.
\end{rem}

\begin{rem}
We see that, for each sub-symmetry defined by condition (\ref{sym_id}) with corresponding 
differential functions $\Xi^v[u]$, there exists an associated conservation law (\ref{con}).  
These results are similar to those of \cite{Rosenhaus94} and \cite{Rosenhaus02}, 
which were obtained in the frame of symmetries. As noted above, 
see also \cite{RS2016}, the conditions for sub-symmetries are 
weaker than those for symmetries and do not require the invariance of all equations of the system.
\end{rem} 

A useful special case of Theorem \ref{thm:subcon} is related to vector field deformations of a conservation law discussed in \cite{RS2016} (Theorem 3): 
sub-symmetry deformations of a conservation law lead to further conservation laws.
If $\Xi^v\Delta_v=D_iM^i$ is a conservation law, and $(\Xe,\Xi^v\Delta_v)$ is a sub-symmetry of the system \eqref{sys},   
then $\Xe(\Xi^v\Delta_v)=D_i(\Xe\,M^i)$ is also a conservation law. We will present this result in a slightly different form:
\begin{cor}[Total divergences]										\label{cor:total_div}
Suppose the system \eqref{sys} has a conservation law $\Xi^v\Delta_v=D_iM^i$ with some smooth functions $M^i[u],\,\, 1\le i\le p$. 
Then for each sub-symmetry $(\Xe,\Xi^v\Delta_v)$,
the p-tuple $\tilde M^i=\text{\emph{X}}M^i$ 
is a conservation law of \eqref{sys} satisfying 
the following continuity equation:
\begin{align}\label{con_div}
D_i\tilde M^i=\Lambda^v\Delta_v=0
\end{align}
when $\Delta=0$.
\end{cor}
\begin{proof}
Since $E_a\left(\Xi^v\Delta_v\right)= 0$ in this case, then according to \eqref{proof1}, we have
$\Gamma^{av}=0$.  By the Noether identity (\ref{Rosenhaus:equation7}), we have 
$\text{X}D_iM^i=D_iR^i(\Xi^v\Delta_v)$. However, $\text{X}D_iM^i=D_i\text{X}M^i$
\cite{Olver}, so (\ref{con_div}) follows from (\ref{con}).
\end{proof}

This result is used in the following sections. 
\smallskip

\section{Two-dimensional vorticity equations}
\label{sec:2d}
We consider the incompressible inviscid Euler equations in the vorticity formulation.  It is well known that two-dimensional Cartesian planar flow possesses infinite conservation laws.  For this type of flow, there exists an alternative Hamiltonian formulation of the vorticity equations \cite{Olver82}
where the infinite conservation laws arise due to a degeneracy in the Poisson bracket, and it was 
conjectured 
that there are no symmetries that give rise to the infinite conservation laws in two dimensions.  In Lagrangian coordinates, these conservation laws arise from a ``particle relabeling symmetry" \cite{Salmon88}, but this symmetry connection cannot be identified in Eulerian coordinates. The conserved densities in this theory (``Casimir invariants" or ``isovortical constraints") have important theoretical and practical applications; 
e.g., in conjunction with nonlinear stability analysis, they can be used to prove the formation of macro-structures in two-dimensional turbulent flows (see e.g., \cite{Bouchet}).   

In the paper, we study the relation between these infinite conservation laws and symmetry properties of the system. We show that the infinite conservation laws of the Euler equations in vorticity formulation arise from its infinite sub-symmetries. This is not the only approach to study conservation laws; we could find a Hamiltonian of a corresponding system, and look for its Casimir invariants, or try to develop a Lagrangian coordinate formalism for related differential system, or use the direct method for finding conservation laws for a fluid, see, e.g. \cite{Kelbin}, or \cite{Chevober}.
 
Our approach for finding sub-symmetries and corresponding local conservation laws is quite straightforward and similar to that for finding Lie symmetries. In addition, knowledge of the symmetry structure of a differential system allows one to find a good deal of valuable information about the system, including a geometric structure of the solution manifold, possible linearization, alternative Lagrangians, properties of solutions, etc.

\subsection{Infinite sub-symmetries and infinite conservation laws}
We write the system of Euler and vorticity equations in two dimensions in the form:
\begin{align}											\label{2d:sys}
\begin{split}
&\Delta_1=u^1_t+u^1u^1_1+u^2u^1_2+p_1=0,\\
&\Delta_2=u^2_t+u^1u^2_1+u^2u^2_2+p_2=0,\\
&\Delta_3=\omega_t+u^1\omega_1+u^2\omega_2=0,\\
&\Delta_4=u^1_1+u^2_2=0,\\
&\Delta_5=\omega-(u^2_1-u^1_2)=0.
\end{split}
\end{align}
Here, $\left(\, u^1(x^1,x^2, t), \,\, u^2(x^1,x^2,t),\,\, 0\right)$ are the components of the 2-D fluid velocity vector, $p=p(x^1,x^2,t)$  is the fluid pressure, and $\omega(x^1,x^2,t)=u^2_1-u^1_2$   is the two-dimensional 
vorticity; 
$\omega_t \, = \, \frac{\partial \omega}{\partial t}$,     
$\omega_i = \, \frac{\partial \omega}{\partial x^i}$,   $i=1,2.$ 
Note that for consistency, to the regular Euler equations for velocities we added the equation defining vorticity; this approach is similar to constructing the ideal of exterior forms characterizing the entire differential system, e.g. \cite{Estabrook}. 
 Note also that the system 
\eqref{2d:sys} is quasi-Noether. Indeed,
\begin{align}											\label{2d:quasi1}
\Delta_3+\omega\Delta_4 = D_t(\omega)+ D_1(u^1 \omega) + D_2(u^2 \omega),
\end{align}
and therefore $E_1(\Delta_3+\omega\Delta_4)= E_2(\Delta_3+\omega\Delta_4) =0$.
\begin{prop}											\label{2d:prop:subsym}
The two-dimensional vorticity system \emph{(\ref{2d:sys})}  
admits the following infinite sub-symmetry $(X, \Delta_3+\omega\Delta_4)$:
\begin{align}											\label{2d:subsym}
\text{\emph{X}}\left[\Delta_3+\omega\Delta_4\right]=0
\end{align}
when $\Delta=0$.  The sub-symmetry operator $\text{\emph{X}}$ is given by:
\begin{align}											\label{2d:op}
\text{\emph{X}}=f(\omega)\partial_\omega,
\end{align}
where $f$ is an arbitrary function of $\omega$.
\end{prop}

\begin{prop}											\label{2d:prop:con}
The two-dimensional vorticity system \emph{(\ref{2d:sys})} admits the following infinite 
conservation law:
\begin{align}											\label{2d:con}
D_tf(\omega)+D_j\left[u^jf(\omega)\right]=0
\end{align}
when $\Delta=0$, where $f$ is an arbitrary function of $\omega$.
\end{prop}

Both Propositions are quite transparent.

\subsection{Proofs}
\begin{proof}[Proof of Proposition \emph{\ref{2d:prop:subsym}}]
The following operator identities can be easily verified using (\ref{2d:op}):
\begin{align*}
\text{X}\Delta_3=f'(\omega)\Delta_3,\:\:\:\:\:
\text{X}\Delta_1=\text{X}\Delta_2=\text{X}\Delta_4\,=\,0.
\end{align*}
Therefore 
\begin{align*}
\text{X}\left[\Delta_3+\omega\Delta_4\right]=f'(\omega)\Delta_3+f(\omega)\Delta_4=0
\end{align*}
when $\Delta=0$.  This establishes (\ref{2d:subsym}).
\end{proof}

Observe that $\text{X}\Delta_5=f(\omega)\neq 0$ when $\Delta=0$.  Thus, the operator (\ref{2d:op}) is \textit{not a symmetry} of the Euler/vorticity system.  However, it is a sub-symmetry of this system, which allows the application of Theorem \ref{thm:subcon}.

\begin{proof}[Proof of Proposition \emph{\ref{2d:prop:con}}]
Observe that the following equation is a total divergence:
\begin{align}											\label{2d:div}
\Delta_3+\omega\Delta_4= 
D_t\omega+D_j(u^j\omega).
\end{align}
Using Corollary \ref{cor:total_div} and expressions \eqref{2d:div} and (\ref{con_div}), 
we obtain
\begin{align}											\label{2d:Xdiv}
\text{X}\left(\Delta_3+\omega\Delta_4\right) = 
\text{X}\left(D_t\omega+D_j(u^j\omega)\right)=
D_t\left(f(\omega)\right) + D_j\left(u^j f(\omega)\right).
\end{align}
\noindent According to Proposition \ref{2d:prop:subsym}, the LHS of \eqref{2d:Xdiv} 
vanishes on $\Delta=0$, and we obtain the conservation laws (\ref{2d:con}).
\end{proof}

\noindent In this example, we were able to generate a conservation law of the system 
\eqref{2d:sys} directly from its sub-symmetry \eqref{2d:subsym} using the facts
that it is a quasi-Noether system
and
that the combination 
$\Delta_3+\omega\Delta_4$ is a total divergence. 

We have demonstrated that the Euler equations of two-dimensional incompressible flow possess an infinite sub-symmetry with an arbitrary function of the dependent variable (vorticity), and that this sub-symmetry leads to the well-known infinite series of conservation laws. Physically, this symmetry might be interpreted as resulting from the $z$ direction being a so-called ``invariant direction" corresponding to $\vec\omega=\omega\hat{z}$.  Projected along this direction, the vorticity equation becomes linear and homogeneous in first derivatives, which results in the existence of infinite sub-symmetries and conservation laws.

\section{Constrained Euler equations}
\label{sec:euler}
The constrained three-dimensional Euler equations, in general, involve a constraint on the components of the velocity vector, for example, $\partial_zu^z=0$, where $u^z$ is the $z$ component of the velocity vector $\vec u$.  The crucial distinction between this type of flow and the two-dimensional flows of Section \ref{sec:2d} is that the constrained velocity component $u^z$
here
is nonzero, so the system still has three degrees of freedom;  see \cite{Niet} for examples of such systems and a numerical solution for an azimuthally constrained case.  Some other physically relevant examples occur in helically invariant flows,
considered in \cite{Kelbin}.

Infinite conservation laws analogous to the Casimir invariants of the vorticity equation in two-dimensional flow can also be found for the three-dimensional Euler equations with constraints.  In this section, we consider only conservation laws involving velocities; the vorticity equations are considered in a subsequent section.  Numerous Casimir-like conservation laws were discovered in \cite{Kelbin} for the helically invariant Euler equations. Our results extend theirs in several ways.  First, we obtain general results in vector form without introducing a special helical coordinate system. 
Second, we show that the conservation laws arise from specific sub-symmetries of the constrained Eulerian system and not from a ``relabeling symmetry" (that only holds in Lagrangian coordinates). 
Third, we generate several new families of infinite conservation laws and classify the types of systems that admit such laws.

\subsection{Infinite sub-symmetries and infinite conservation laws}
Consider the constrained Euler system $\Delta_v,\,1\le v\le 6$:
\begin{align}													\label{euler:sys}
\begin{split}
&\Delta_i=u^i_t+u^ju^i_j+p_i=0,\hspace{4 mm}i=1,2,3,\\
&\Delta_4=u^i_i=0,\\
&\Delta_5=a^ip_i=0,\\
&\Delta_6=b^j\partial_j\left(a^iu^i\right)=0,
\end{split}
\end{align}
where $\vec a=\vec\alpha+\vec\beta\times\vec r$,
with constant and otherwise arbitrary vectors $\vec\alpha$ and $\vec\beta$, and $\vec b=\vec a/|\vec a|^2+\vec a\times\vec B$, where $\vec B(t,\vec r)$ is an arbitrary solution of $\vec\nabla\cdot(\vec a\times\vec B)=0$.  Here, $u^i(t,\vec r),\, i=1,2,3$ is the velocity vector, and $p(t,\vec r)$ is the pressure.

\begin{prop}\label{euler:prop:subsym}
The constrained Euler system \emph{(\ref{euler:sys})}  admits the following infinite sub-symmetry:
\begin{align}\label{euler:subsym1}
&\text{\emph{X}}\left[a^i\Delta_i+\left(a^iu^i\right)\Delta_4-\Delta_5\right]=0
\end{align}
when $\Delta=0$. The operator $\text{\emph{X}}$ is given by:
\begin{align}													\label{euler:op}
\text{\emph{X}}=f\left(\vec a\cdot\vec u\right)\,b^i\partial_{u^i},
\end{align}
where $f$ is an arbitrary function.
\end{prop}

\begin{prop}\label{euler:prop:con}
The constrained Euler system \emph{(\ref{euler:sys})} admits the following infinite conservation laws:
\begin{align}\label{euler:con1}
&D_tf\left(a^iu^i\right)+ D_j\left[u^jf\left(a^iu^i\right)\right]=0,
\end{align}
when $\Delta=0$, where $f$ is an arbitrary function.
\end{prop}

\subsection{Construction}
To prove Propositions \ref{euler:prop:subsym}--\ref{euler:prop:con}, we start with the ``unconstrained" Euler system and derive constraints on the Euler equations from the condition of the existence of infinite sub-symmetries and conservation laws. The idea is to find some direction vector $\vec a(t,\vec r)$ along which the (nonlinear) Euler equations become linear in the ``projected variable" $v\equiv\vec a\cdot\vec u=a^iu^i$. The requirement for our system to admit the infinite sub-symmetry $v\to f(v)$ with canonical generator $\text{X}=f(v)\partial_v$ imposes some additional constraints on the vectors ${\vec a}$ and ${\vec b}$, and functions $u^i,\,p$ that define a new ``constrained system".

First, we establish the constrained system (\ref{euler:sys}).  Consider the ``unconstrained" three-dimensional Euler system:
\begin{align}\label{euler:sys2}
\begin{split}
&\Delta_i=u^i_t+u^ju^i_j+p_i=0,\hspace{4 mm}i=1,2,3,\\
&\Delta_4=u^i_i=0.
\end{split}
\end{align}
We examine the projection of $\Delta_i$ along some vector $a^i$ and rewrite this 
equation in terms of the ``projected variable" $v=a^iu^i$:
\begin{align}\label{euler:proj}
\begin{split}
&a^i\Delta_i=a^iu^i_t+u^ju^i_ja^i+a^ip_i=0\\
&\hspace{8 mm}=v_t+u^jv_j-u^i(a^i_t+u^ja^i_j)+a^ip_i.
\end{split}
\end{align}

In order to admit the infinite sub-symmetry $v\to f(v)$, this equation has to be linear and 
homogeneous in $v$. Requiring that the last three terms vanish, we obtain:
\begin{align}												\label{euler:constraint}
\begin{split}
&\Delta_5=a^ip_i=0,\\
&a^i_t=0,\hspace{4 mm}1\le i\le 3,\\
&u^ia^i_ju^j=0.
\end{split}
\end{align}
The solution of the last two conditions is: 
\begin{align}												\label{euler:vector_a}
\vec a=\vec\alpha+\vec\beta\times\vec r
\end{align}
with some constant vectors $\vec\alpha,\vec\beta$.  Note that our projected variable $v=\vec a\cdot\vec u$ is a linear combination of linear 
and angular momenta. Now, (\ref{euler:proj}) becomes:
\begin{align}												\label{euler:proj1}
a^i\Delta_i-\Delta_5=v_t+u^jv_j.
\end{align}
This equation is still nonlinear in $v$ ($v=a^iu^i$) and does not admit $v\to f(v)$.
Applying the sub-symmetry operator $\text{X}=f(v)\partial_v$ to this equation, we obtain:
\begin{align*}
\text{X}(a^i\Delta_i-\Delta_5)=\partial_tf+u^j\partial_jf+v_j\text{X}u^j=f'(v)(a^i\Delta_i-\Delta_5)+fv_j\partial_v u^j.
\end{align*}
Requiring that the last (nonlinear) term vanish, we get: 
\begin{align}\label{euler:cond}
\frac{\partial u^j}{\partial v}\frac{\partial v}{\partial x^j}=0.
\end{align}
Let us define an invertible transformation $(u^1,u^2,u^3)\to(v^1,v^2,v^3)$ under the assumption that $v^1(t,\vec r,\vec u)=v=a^iu^i$.  Since the transformation is invertible, we can write $u^i=U^i(t,\vec r,\vec v)$ for some differentiable functions $U^i$.  The constraint (\ref{euler:cond}) becomes:
\begin{align}\label{euler:cond2}
\frac{\partial U^j}{\partial v^1}\frac{\partial v^1}{\partial x^j}=0.
\end{align}
There are two conditions upon the functions $U^i(t,\vec r,\vec v)$: that their Jacobian be nonzero and finite, and that $a^iU^i=v^1$, or 
$a^i\partial_{v^1}U^i=1$.  A necessary condition for the former constraint is that $\partial_{v^1}U^i$ be nonzero for some $i$; all other partial derivatives of $U^i$ can be arbitrarily chosen (subject to second derivative compatibility and the nonsingularity of the Jacobian). Introducing 
$\vec b=\vec b(t,\vec r,\vec u)=\partial_{v^1}\vec U(t,\vec r,\vec v))$ we see that $\vec b$ is a differentiable vector that satisfies only two conditions: $|\vec b|^2>0$ and $\vec a\cdot\vec b=1$. The constraint \eqref{euler:cond} (\eqref{euler:cond2}) takes the form:
\begin{align}\label{euler:constraint2}
\Delta_6=b^j\partial_j(a^iu^i)=0,
\end{align}
where vector $\vec b$ can be represented as:
\begin{align*}
\vec b=\frac{1}{|\vec a|^2}\vec a+\vec a\times\vec B
\end{align*}
with an arbitrary vector $\vec B(t,\vec r)$ (we assumed for simplicity that $\partial_{u^i}\vec B\equiv 0$). 

In terms of $\vec b$, the sub-symmetry $\text{X}$ becomes (see \eqref{euler:op}):
\begin{align*}
\text{X}=f(v)\partial_v=f(v^1)\frac{\partial U^j}{\partial v^1}\partial_{u^j}=f(a^iu^i)b^j\partial_{u^j}.
\end{align*}
Now we have our constrained system: (\ref{euler:sys2}), (\ref{euler:constraint}), and (\ref{euler:constraint2}), 
and we can compute the action of the operator \eqref{euler:op} on the combination $(a^i\Delta_i-\Delta_5)$ using
\eqref{euler:proj1}:
\begin{align*}
\text{X}(a^i\Delta_i-\Delta_5)=\partial_tf+u^j\partial_jf+v_jb^jf=f'(a^i\Delta_i-\Delta_5)+f\Delta_6=0
\end{align*}
when $\Delta=0$.  We see that $\text{X}$ is a sub-symmetry of $\Delta$.  
However, according to \eqref{euler:proj1} this combination is not a total divergence. 
Consider instead
\begin{align}											\label{euler:combination}
a^i\Delta_i-\Delta_5+a^iu^i\Delta_4=v_t+\partial_j(u^jv).
\end{align}
and compute the action of $\text{X}$ on it:
\begin{align}\label{euler:action}
\begin{split}
&\text{X}(a^i\Delta_i-\Delta_5+a^iu^i\Delta_4)=D_tf + D_j(u^jf+vfb^j)\\
&\hspace{41 mm}=f'(a^i\Delta_i-\Delta_5)+f\Delta_4+(f+vf')\Delta_6+vfb^j_j.
\end{split}
\end{align}
Evidently, $\text{X}$ is a sub-symmetry if we choose vector $\vec b$ to be divergenceless: $b^j_j=0$. It follows that the vector $\vec B$ should satisfy $\vec\nabla\cdot(\vec a\times\vec B)=0$.  
A particular solution is $\vec B=\vec\beta\times\vec\gamma$, where $\vec\gamma(t)$ is arbitrary. For this $\vec b$, the operator $\text{X}$ is indeed a sub-symmetry of the constrained system \eqref{euler:sys}.  This establishes equation (\ref{euler:subsym1}) of Proposition \ref{euler:prop:subsym}.

Finally, we establish Proposition \ref{euler:prop:con}. Since equation \eqref{euler:combination} is a 
total divergence and admits an infinite sub-symmetry 
(\ref{euler:op}), we can use Corollary \ref{cor:total_div} to generate a corresponding infinite 
set of conservation laws. From (\ref{euler:action}), we obtain:
\begin{align}													\label{euler:con}
D_tf(a^iu^i)+D_j\left(u^jf(a^iu^i)+f(a^iu^i)a^ku^k b^j\right) = 0
\end{align}
when $\Delta=0$.
\noindent The last term in (\ref{euler:con}) can be presented as
\begin{align}
D_j(b^jg)=g'b^jD_j(a^iu^i)=f'\Delta_6 ,\qquad
g(x) = x f(x), \quad x= a^iu^i.
\end{align}
\noindent The remaining infinite conservation law has the form:
\begin{align*}													
D_tf(a^iu^i)+D_j\left(u^jf(a^iu^i)\right) = 0, 
\end{align*}
\noindent when $\Delta=0$ (function $f$ is arbitrary). This establishes equation (\ref{euler:con1}) of Proposition \ref{euler:prop:con}.

\subsection{Discussion}

According to Proposition \ref{euler:prop:con}, the constrained Euler system possesses an infinite set of conservation laws parametrised by arbitrary functions of the velocity $\vec u$ projected along the vector $\vec a$:  $\vec a=\vec\alpha+\vec\beta\times\vec r$, with arbitrary constant vectors $\vec\alpha$, and $\vec\beta$.  Thus, there exist certain ``invariant" directions in space with direction vector $\vec a$ along which the constrained Euler system possesses an infinite sub-symmetry and an infinite set of conservation laws.  Given that the vector $\vec a$ is a linear combination of infinitesimal translations and rotations, under which the Euler equations are symmetric, it is an interesting question as to if these invariant directions are related to the equations' symmetry properties.

The infinite series of conservation laws \eqref{euler:con1} and its connection to infinite sub-symmetries \eqref{euler:subsym1} is (to the best of our knowledge) new. The helical conservation laws reported in \cite{Kelbin} correspond to the case when $\vec\alpha=-b\hat{e}_z$, $\vec\beta=a\hat{e}_z$, and $\vec a=-ay\hat{e}_x+ax\hat{e}_y-b\hat{e}_z$, (see (\ref{euler:vector_a})) where $\hat{e}_z$ is the Cartesian unit basis vector in the $z$ direction. The infinite series of conservation laws of generalized momenta/angular momenta ((4.9) in \cite{Kelbin}) can be obtained from \eqref{euler:con1} by choosing $ a^2=-a^1 = a,\,\,\, a^3 =b$.

There are, potentially, some interesting applications of the 
results related to the existence of infinite sub-symmetries and infinite conservation laws for the Euler system.  For example, it allows for the investigation of flow stability and turbulence structure in three dimensions (subject to a flow constraint) using the nonlinear stability analysis developed by Arnold, see e.g. \cite{Arnold1966}.  In addition, it suggests the possibility that, for these flow configurations, the Euler equations might be linearizable via a hodograph transformation.
\smallskip

Note that the very existence of infinite sub-symmetries \eqref{euler:subsym1} (and corresponding infinite conservation laws \eqref{euler:con1}) for the essentially non-linear system of Euler equations is quite remarkable. And the fact that we considered a constrained Euler system \eqref{euler:sys} does not change this conclusion; for regular symmetries, the existence of an additional ``conditional" symmetry for a constrained system is usually, rather limited, see e.g. \cite{Levi1989}, \cite{Olver87}.

\subsection{Less Constrained Euler system}
The result (\ref{euler:con1}) can also be obtained without imposing $\Delta_6$ in \eqref{euler:sys}. Specifically, we consider the following constrained Euler system:
\begin{align}\label{euler:wsys}
\begin{split}
&\Delta_i=u^i_t+u^ju^i_j+p_i=0,\hspace{4 mm}i=1,2,3,\\
&\Delta_4=u^i_i=0,\\
&\Delta_5=a^ip_i=0,
\end{split}
\end{align}
where $\vec a=\vec\alpha+\vec\beta\times\vec r$ as before.  Note that the only constraint we are imposing ($\Delta_5$) is on the pressure $p$; there is no restriction on the velocities $\vec u$.  

The system (\ref{euler:wsys}) admits the following infinite sub-symmetry operator:
\begin{align}\label{euler:wop}
\text{X}=\frac{1}{|\vec a|^2}f(\vec a\cdot\vec u)a^i\partial_{u^i}-\frac{1}{|\vec a|^2}\vec a\cdot\vec uf(\vec a\cdot\vec u)\partial_p=f(v)\partial_v-\frac{1}{|\vec a|^2}vf(v)\partial_p.
\end{align}
Note that this sub-symmetry is different from (\ref{euler:op}); the extra term involving variations of the pressure $p$ is related to the full nonlinearity of this system.  
Consider 
the following linear combination of equations analogous to \eqref{euler:combination}:
\begin{align}\label{euler:combo}
a^i\Delta_i+a^iu^i\Delta_4=\partial_tv+\partial_j(vu^j)+a^i\partial_ip.
\end{align}
An application of $\text{X}$ to this combination gives:
\begin{align*}
&\text{X}\left(a^i\Delta_i+a^iu^i\Delta_4\right)=D_tf+D_j(u^jf+v\text{X}u^j)-a^i D_i\left(\frac{1}{|\vec a|^2}vf\right)\\
&\hspace{32 mm}=D_tf+D_j(u^jf)\\
&\hspace{32 mm}=f'(a^i\Delta_i-\Delta_5)+f\Delta_4=0
\end{align*}
when $\Delta=0$ ($\vec\nabla\cdot\vec a=0$ for $\vec a=\vec\alpha+\vec\beta\times\vec r$). We see that the operator (\ref{euler:wop}) is indeed, a sub-symmetry of the system (\ref{euler:wsys}),
leading to the infinite conservation law (\ref{euler:con1}).

\section{Constrained vorticity equations}
\label{sec:vort}
\subsection{Infinite sub-symmetries and infinite conservation laws}
Consider the three-dimensional constrained vorticity system $\,\Delta_v,\,\,1\le v\le 12$:
\begin{align}\label{vort:sys}
\begin{split}
&\Delta_{i}=\omega^i_t+u^j\omega^i_j-\omega^ju^i_j+\omega^iu^j_j-u^i\omega^j_j=0,\hspace{4 mm}i=1,2,3\\
&\Delta_4=\omega^i_i=0,\\
&\Delta_{4+i}=u^i_t+u^ju^i_j+p_i=0,\hspace{4 mm}i=1,2,3,\\
&\Delta_8=u^i_i=0,\\
&\Delta_{8+i}=\omega^i-\epsilon^{ijk}u^k_j=0,\hspace{4 mm}i=1,2,3,\\
&\Delta_{12}=\varphi_iu^i-\gamma=0,
\end{split}
\end{align}
where $\varphi_i$
are the components of 
the gradient of some function $\varphi(\vec r)$, $\gamma = \text{const}$, $u^i(t,\vec r),\,$ and $\omega^i(t,\vec r),\,\,i=1,2,3$ are the velocity and vorticity vectors, and $p(t,\vec r)$ is the pressure. Equations 5 -- 8 are Euler equations, equations 9 -- 11 define vorticity vector $\vec\omega=\vec\nabla\times\vec u$, equations 1--3 are the vorticity equations obtained by taking the curl of the Euler equations, and equation 4 is a consequence of equation 8. Equation 12 is the constraint that is discussed below.

\begin{prop}
\label{vort:prop:subsym}
The constrained vorticity system \emph{(\ref{vort:sys})} admits the following infinite sub-symmetry $X$:
\begin{align}
\label{vort:subsym}
{\text{\emph{X}}}\left[\varphi^i\Delta_i+\left(\varphi_iu^i\right)\Delta_4+\omega^jD_j\Delta_{12}\right]=0
\end{align}
when $\Delta=0$.  The sub-symmetry operator ${\text{\emph{X}}}$ is given by:
\begin{align}
\label{vort:op}
{\text{\emph{X}}}=\frac{1}{|\vec\nabla\varphi|^2}f\left(\varphi_i\omega^i\right)\varphi_i\partial_{\omega^i},
\end{align}
where $f$ is an arbitrary function of $\varphi_i\omega^i$.
\end{prop}

\begin{prop}
\label{vort:prop:con}
The constrained vorticity system \emph{(\ref{vort:sys})} admits the following infinite conservation laws:
\begin{align}
\label{vort:con}
D_tf\left(\varphi_i\omega^i\right)+D_j\left[u^jf\left(\varphi_i\omega^i\right)\right]=0
\end{align}
when $\Delta=0$, where $f$ is an arbitrary function of $\varphi_i\omega^i$, 
and $\varphi_i(\vec r)=\frac{\partial \varphi}{\partial x^i}$.
\end{prop}

\begin{prop}
\label{vort:prop:con2}
If $\gamma=0$ in the constrained vorticity system \emph{(\ref{vort:sys})}, then the arbitrary function $f(\varphi_i\omega^i)$ in the above propositions can be replaced by $f(\varphi,\varphi_i\omega^i)$, where $f$ is an arbitrary function of its arguments.
\end{prop}

\subsection{Construction}
We establish Propositions \ref{vort:prop:subsym}-\ref{vort:prop:con2} by a construction
similar to that done for the Euler equations in Section \ref{sec:2d}. However, because the vorticity equation is already linear, our objective here is to make the equation for the projected variable $w=a^i\omega^i$ homogeneous in first derivatives.  This involves placing a constraint on the velocity $\vec u$ rather than on the vorticity $\vec\omega$. 

First, we establish the constrained system (\ref{vort:sys}).  Consider the ``unconstrained" three-dimensional vorticity system: the first 11 equations of the system \eqref{vort:sys}. To construct an appropriate constraint for the vector $\vec a$, we consider a linear combination of the ``main" equations $a^i\Delta_i$ along some direction vector $a^i(t,\vec r),\, i=1,2,3$ and rewrite this combination using the projected variables $v=a^iu^i$ and $w=a^i\omega^i$:
\begin{align}\label{vort:proj}
\begin{split}
&a^i\Delta_i=a^i\omega^i_t+u^j\omega^i_ja^i-\omega^ju^i_ja^i+wu^j_j-v\omega^j_j\\
&\hspace{8 mm}=D_tw+D_j(u^jw)-D_j(\omega^jv)-\omega^ia^i_t-u^ja^i_j\omega^i+\omega^ja^i_ju^i.
\end{split}
\end{align}
We define the vector $\vec a$ in such a way that an infinite sub-symmetry is generated.  The last three terms of (\ref{vort:proj}) vanish if we select
$a^i=\varphi_i(\vec r)=\frac{\partial \varphi}{\partial x^i}$
for some function $\varphi$.
It was shown in \cite{Rosenhaus02} that arbitrary functions $\varphi(\vec r)$ lead only to a finite number of essential (integral) conservation laws (see also \cite{Rosenhaus03a}); therefore, these types of continuity equations are of lesser interest to us.  We fix $\varphi$ in such a way so that an infinite number of conservation laws is generated.

Consider:
\begin{align}\label{vort:proj2}
\varphi_i\Delta_i+v\Delta_4=D_tw+D_j(u^jw)-\omega^jD_j(\varphi_iu^i).
\end{align}
Similar to the situation with equation \eqref{euler:proj1}
making the last term in (\ref{vort:proj2}) vanish    
\begin{align}\label{vort:constraint}
\Delta_{12}=v-\gamma=\varphi_iu^i-\gamma=\vec u\cdot\vec\nabla\varphi(\vec r)-\gamma=0, \quad \gamma=\text{ const},
\end{align}
we anticipate that an infinite sub-symmetry will be obtained.  The intersection of this constraint (\ref{vort:constraint}) and the unconstrained system forms the constrained system (\ref{vort:sys}).

Note that the velocity projection $v$ in this case is required to be constant \textit{everywhere}, while the constraint imposed for 
the Euler equations ($\Delta_6$ in \eqref{euler:sys}) required for $v$ to be constant only in a single direction determined by $\vec a$.

Next, we establish Proposition \ref{vort:prop:subsym}.  By defining $\Delta_{12}$ in (\ref{vort:constraint}), we find that the following conservation law (total divergence equation) for $w$ has been constructed:
\begin{align}\label{vort:proj3}
\varphi_i\Delta_i+v\Delta_4+\omega^jD_j\Delta_{12}=D_tw + D_j(u^jw)=0
\end{align}
when $\Delta_v=0,\,\,1\le v\le 12$.  To show that \eqref{vort:op} is, in fact, a sub-symmetry of (\ref{vort:proj3}) with respect to system (\ref{vort:sys}), we apply the operator $\text{X}$ to (\ref{vort:proj3}) and obtain:
\begin{align} \label{7.10}
&\text{X}\left(\varphi_i\Delta_i+v\Delta_4+\omega^jD_j\Delta_{12}\right)=D_tf(w)+D_j(u^jf(w))\\ 
&\hspace{50 mm}=f'(w)(D_tw+u^jD_jw)+f(w)\Delta_4 \nonumber  \\
&\hspace{50 mm}=f'(w)\left[\varphi_i\Delta_i+v\Delta_4+\omega^jD_j\Delta_{12}-w\Delta_8\right]+f(w)\Delta_4=0 \nonumber
\end{align}
when $\Delta=0$.  This establishes Proposition \ref{vort:prop:subsym}.

According to the expression \eqref{7.10}, the sub-symmetry $\text{X}$ gives rise to the infinite conservation law: 
\begin{align*}
D_tf(w)+D_j\left(u^jf(w)\right)=0
\end{align*}
when $\Delta=0$, where $f$ is an arbitrary function, and $w=\varphi_i\omega^i$.  This establishes Proposition \ref{vort:prop:con}.  Proposition \ref{vort:prop:con2} follows from virtually the same computations.

\subsection{Discussion}

The infinite series of conservation laws \eqref{vort:con}, and its connection to infinite sub-symmetries \eqref{vort:subsym} is new. There are some notable special cases.  

1. Zero flow velocity in the $z$ direction (i.e. $\hat{e}_z\cdot\vec u=0$) and no dependence of the dynamical variables on $z$ (Cartesian plane). This configuration corresponds to the well-known case of
two-dimensional vorticity equations
discussed in Section \ref{sec:2d}. The infinite conservation laws \eqref{2d:con} can be obtained from (\ref{vort:con}) by setting $\varphi(\vec r)=z$.
In this case, the infinite conservation laws contain arbitrary functions of $w$, $\,w=\hat{e}_z\cdot\vec\omega=\omega_z$.

2.
The case of
``cylindrical invariance" discussed in \cite{Kelbin} as a special case of their helically invariant system (see (6.30) in \cite{Kelbin}).
In this case, we have $\varphi(\vec r)=\arctan\left(y/x\right),$ the azimuthal angle; $r=\sqrt{x^2+y^2}$, and it follows that infinite conservation laws in this case contain arbitrary functions of $w=\left(-\frac{y}{r^2}\hat{e}_x+\frac{x}{r^2}\hat{e}_y\right)\cdot\vec\omega=\frac{1}{r}\omega_\varphi$.

Note that the helical invariance discussed in \cite{Kelbin} does not fall under this framework since there exists no function $\varphi$ which determines the helical invariance of the velocity field.
However, within our framework, it is possible to generate infinite conservation laws associated with velocity constraints other than axial and cylindrical. For example, if the velocity constraint is spherical, then the system will have infinite vorticity conservation laws associated with a spherical sub-symmetry. Similar results would also hold for oblate spheroidal invariance, paraboloidal invariance, elliptical invariance, etc.

\subsection{Non-triviality of infinite conservation laws}
The vorticity system (\ref{vort:sys}) is, obviously, underdetermined since some of the equations are differential consequences of the others. 
For example, $D_i \Delta_{8+i} = \Delta_4$ everywhere in the jet space.  Thus, in order to make sure that the conservation laws for such an abnormal 
system (cf. discussion in \cite{Olver}) are non-trivial, we need to account for all possible interdependencies among the equations of the system.
We examine the characteristics of the conservation law \eqref{vort:con}:
\begin{align}\label{vort:characteristic}
&D_tf(w)+D_j\left(u^jf(w)\right) \nonumber \\
&=f'(w)\left[\varphi_i\Delta_i+v\Delta_4+\omega^jD_j\Delta_{12}-w\Delta_8\right]+f(w)\Delta_4 \\ 
&=f'(w)\varphi_i\Delta_i+\left[\gamma f'(w)+f(w)\right]\Delta_4-wf'(w)\Delta_8-\omega^jD_jf'(w)\Delta_{12}+D_j\left(f'(w)\omega^j\Delta_{12}\right). 
\nonumber
\end{align}
The last term is a trivial conservation law of the first kind. Notice that $\Delta_{12}$, and $\Delta_8$  have no interdependencies with other equations of the subset. Rewriting the vorticity equations as $\Delta_i=\omega^i_t+\left[\vec\nabla\times(\vec\omega\times\vec u)\right]^i,\,i=1,2,3$, we can see that $D_i\Delta_i=\omega^i_{it}=D_t\Delta_4$, and the following differential relationship:
\begin{align}\label{vort:inter}
D_t\Delta_4-D_i\Delta_i=0
\end{align} 
holds everywhere in the jet space. Using the interdependency relation (\ref{vort:inter}) we can rewrite equation (\ref{vort:characteristic})
as follows:
\begin{align*}
&f'(w)\varphi_i\Delta_i=D_i\left(f'(w)\varphi\Delta_i\right)-\varphi D_if'(w)\Delta_i-\varphi f'(w)D_i\Delta_i\\
&=D_i\left(f'(w)\varphi\,\Delta_i\right)-\varphi D_if'(w)\Delta_i-\varphi f'(w)D_t\Delta_4\\
&=\varphi\, D_tf'(w)\Delta_4-\varphi D_if'(w)\Delta_i-D_t\left(f'(w)\varphi\Delta_4\right)+D_i\left(f'(w)\varphi\Delta_i\right)
\end{align*}
Substituting into (\ref{vort:characteristic}) and discarding the trivial conservation laws gives for the rhs:
\begin{align*}
-\varphi\,D_if'(w)\Delta_i+\left[\gamma f'(w)+f(w)+\varphi \partial_tf'(w)\right]\Delta_4-wf'(w)\Delta_8-\omega^j D_jf'(w)\Delta_{12}.
\end{align*}
We can see that the characteristics of the conservation law (\ref{vort:characteristic}) (coefficients of the $\Delta$'s) are nonzero in this expression, and there is no interdependencies between remaining equations in the rhs of (\ref{vort:characteristic}). Therefore, 
the infinite conservation law (\ref{vort:con}) is nontrivial.

\subsection{Infinite sub-symmetries and conservation laws for helically invariant Euler equations}
\label{subsec:2comp}
The two-component helically invariant Euler and vorticity equations are as follows \cite{Kelbin}:
\begin{align}\label{vort:helicalsys}
\begin{split}
&\Delta_1=\partial_tu^r+u^r\partial_ru^r+\frac{1}{B}u^\xi\partial_\xi u^r-\frac{b^2B^2}{r^3}\left(u^\xi\right)^2+\partial_rp=0\\
&\Delta_2=\partial_tu^\xi+u^r\partial_ru^\xi+\frac{1}{B}u^\xi\partial_\xi u^\xi+\frac{b^2B^2}{r^3}u^ru^\xi+\frac{1}{B}\partial_\xi p=0\\
&\Delta_3=\frac{1}{r}u^r+\partial_ru^r+\frac{1}{B}\partial_\xi u^\xi=0\\
&\Delta_4=\omega^\eta-\left[\frac{1}{B}\partial_\xi u^r-\frac{1}{r}\partial_r\left(ru^\xi\right)+\frac{a^2B^2}{r}u^\xi\right]=0\\
&\Delta_5=\partial_t\omega^\eta+\frac{1}{r}\partial_r\left(ru^r\omega^\eta\right)+\frac{1}{B}\partial_\xi\left(u^\xi\omega^\eta\right)-\frac{a^2B^2}{r}u^r\omega^\eta=0,
\end{split}
\end{align}
where $B(r)=r/\sqrt{a^2r^2+b^2}$
($a$, $b$ are constants).  Functions $u^r,u^\xi,$ and $\omega^\eta$ of $(t,r,\xi)$ represent projections of the velocity vector $\vec u$ and vorticity vector $\vec\omega$ with respect to the orthonormal basis $(\hat{e}_r,\hat{e}_\xi,\hat{e}_{\perp\eta})$ and orthogonal coordinates $r=\sqrt{x^2+y^2}$ and $\xi=az+b\arctan\left(y/x\right)$.  In these variables, a total divergence of a 3-tuple $(M^t, M^r, M^\xi)$ is
given by:
\begin{align}
D_tM^t+\frac{1}{r}D_r\left(rM^r\right)+\frac{1}{B}D_\xi M^\xi =0.
\end{align}
Note that for these flows, there is no dependence on helical variable $\eta=a\arctan(y/x)-bz/r^2$, and the third velocity component $u^\eta=\hat e_{\perp\eta}\cdot\vec u=0$.

Let us show that $(X, \: \frac{B}{r}\Delta_5)$ is an infinite sub-symmetry of the system \eqref{vort:helicalsys}, where 
\begin{align}\label{vort:helicalsym}
\text{X}=f(w)\partial_w=\frac{r}{B}f\left(\frac{B}{r}\omega^\eta\right)\partial_{\omega^\eta},
\end{align}
$f$ is an arbitrary function and $w=B\omega^\eta/r$. 

\noindent Indeed, notice first that the combination $\frac{B}{r}\Delta_5$ is a total divergence:
\begin{align}\label{vort:helicalcombo}
\begin{split}
&\frac{B}{r}\Delta_5=D_t\left(\frac{B}{r}\omega^\eta\right)+\frac{1}{r}D_r\left(Bu^r\omega^\eta\right)+
\frac{1}{B} D_\xi\left(\frac{B}{r}u^\xi \omega^\eta\right)-\frac{a^2B^3}{r^2}u^r\omega^\eta-
D_r\left(\frac{B}{r}\right)u^r\omega^\eta\\
&\hspace{9 mm}=D_tw+\frac{1}{r}D_r\left(ru^rw\right)+\frac{1}{B}D_\xi\left(u^\xi w\right),
\end{split}
\end{align}
The action of operator (\ref{vort:helicalsym}) on this combination is:
\begin{align}\label{vort:helicalaction}
\begin{split}
&\text{X}\left(\frac{B}{r}\Delta_5\right)= D_tf(w)+\frac{1}{r} D_r\left(ru^rf(w)\right)+\frac{1}{B}D_\xi\left(u^\xi f(w)\right)\\
&\hspace{19 mm}=f'(w)\left[\frac{B}{r}\Delta_5-w\Delta_3\right]+f(w)\Delta_3=0
\end{split}
\end{align}
when $\Delta=0$.  Therefore, (\ref{vort:helicalsym}) is an infinite sub-symmetry of $\frac{B}{r}\Delta_5$ of the system (\ref{vort:helicalsys}). Note that the operator $X$ is not a symmetry of the vorticity system (\ref{vort:helicalsys}), since $\text{X}\Delta_4\neq 0$ when $\Delta=0$.

\smallskip
The equation \eqref{vort:helicalaction} generates an infinite series of vorticity conservation laws $(M^t, M^r, M^\xi)$: 
\begin{align}												\label{Kelbin:cons_omega}
M^t = f\left(\frac{B}{r}\omega^\eta\right), \quad
M^r = u^r f\left(\frac{B}{r}\omega^\eta\right), \quad
M^{\xi} = u^{\xi} f\left(\frac{B}{r}\omega^\eta\right),
\end{align}
with an arbitrary function $f$, which was obtained in \cite{Kelbin}.

\bigskip

We could also consider a more general three-component helically invariant flow with $\hat e_{\perp\eta}\cdot\vec u\neq 0$ , and complement our system with additional helical constraints as in \cite{Kelbin}. The resulting constrained Euler vorticity system can be shown to have two additional infinite sub-symmetries $(X_1,\Delta_h)$ and $(X_2,\Delta_e)$, and conservation laws:  
\begin{align}
\label{hel2:op1}
&\text{X}_1=H\omega^i\partial_{\omega^i}+HE\partial_p, \nonumber \\
&\text{X}_2=\frac{1}{|\vec a|^2}Ha^i\partial_{u^i}-\frac{1}{|\vec a|^2}H\vec a\cdot\vec u\,
\partial_p=\frac{B}{r}H\partial_{u^\eta}-
\frac{B}{r}
Hu^\eta \partial_{p}, \nonumber \\
&\Delta_h=\partial_th+\vec\nabla\cdot\left(\vec u\times\vec\nabla E+(\vec\omega\times\vec u)\times\vec u\right)=0,\\
&\Delta_{e}=\vec\nabla\cdot\left[E\left(\vec u-u^\eta e_{\perp\eta}\right)\times\vec\nabla\left(\frac{r}{B}u^\eta\right)\right]
\nonumber =0
\end{align}
when $\Delta=0$. Here $h=\vec u\cdot\vec \omega$ is the helicity, $\Delta_h=0$ is the helicity conservation law, $E=\frac{1}{2}|\vec u|^2+p$ is the energy, $H=H(\vec a\cdot \vec u)$ is an arbitrary function, and $\vec a=a(-ye_x+xe_y)-be_z$. 

A combination of these two infinite sub-symmetries 
yields the generalized helicity conservation law obtained in \cite{Kelbin}:
\begin{align} \label{gener_sub}
X_1\Delta_h-X_2\Delta_{e}=
\partial_t\left(\vec u\cdot\vec\omega\,H\right)+\vec\nabla\cdot\left[H\left(\vec u\times\vec\nabla E+(\vec\omega\times\vec u)\times\vec u\right)+Eu^\eta e_{\perp\eta}\times\vec\nabla H\right]=0,
\end{align}
(when $\Delta=0$).

\smallskip

\subsection{Helicity sub-symmetry and conservation law}
Although conservation laws for Lagrangian systems can always be obtained from corresponding (variational) symmetries, this may not be the case for non-Lagrangian systems.  The helicity conservation law of the Euler equations is one such example. In \cite{Olver82}, a Hamiltonian structure of the vorticity equations was introduced, and conservation laws corresponding to (Hamiltonian) point symmetries were constructed using Noether's theorem. In this way, conservation of energy was obtained from invariance under time translations:
\begin{align*}
\partial_{t}\left(\frac{1}{2}|\vec u|^2\right)+\vec\nabla\cdot\left[\vec u\left(\frac{1}{2}|\vec u|^2+p\right)\right]=0.
\end{align*}
However, conservation of helicity
\begin{align}\label{helic_cons}
\partial_t\left(\vec u\cdot\vec\omega\right)+\vec\nabla\cdot\left[\vec u\times\vec\nabla E+\left(\vec\omega\times\vec u\right)\times\vec u\right]=0
\end{align}
and the infinite ``Casimir" conservation laws discussed in Section \ref{sec:2d} could not be found from point symmetries of the system.  Instead, in \cite{Olver82}, this conservation law (and the infinite vorticity conservation laws) were related to the degeneracy of the Hamiltonian Poisson bracket.  

We show that the helicity conservation law \eqref{helic_cons} can be naturally obtained from a sub-symmetry of the Euler system.  Consider the Euler equations:
\begin{align}\label{helicity:sys}
\begin{split}
&\Delta_i=u^i_t+u^ju^i_j+p_i=0,\hspace{4 mm}i=1,2,3,\\
&\Delta_4=u^i_i=0.
\end{split}
\end{align}
The conservation of energy of the system \eqref{helicity:sys}:
\begin{align}\label{helicity:energy}
\Delta_E=\partial_{t}\left(\frac{1}{2}u^iu^i\right)+\partial_j\left[Eu^j\right]=u^i\Delta_i+E\Delta_4=0, \qquad E=\frac{1}{2}u^iu^i+p,
\end{align}
when $\Delta=0$.  Because this equation is a total divergence, it satisfies the condition (\ref{euler_id}).  With this in mind, consider the following infinitesimal transformation:
\begin{align}\label{helicity:sym}
X_h=\left[\vec\omega+\frac{1}{E}\left(\vec\omega\times\vec u\right)\times\vec u\right]\cdot\partial_{\vec u}-\left(\vec u\cdot\vec\omega\right)\partial_p,
\end{align}
where $\vec\omega=\vec\nabla\times\vec u$ and $\vec a\cdot\partial_{\vec u}= a^i\partial_{u^i}$. 
\smallskip

Let us show that $(X_h, \Delta_E)$ 
is a sub-symmetry of the Euler system (\ref{helicity:sys}). Indeed, calculating $X_h\Delta_E$ and using Corollary \ref{cor:total_div}, we obtain 
\begin{align}\label{helicity:action}
X_h\Delta_E=\partial_t\left(\vec u\cdot\vec\omega\right)+\vec\nabla\cdot\left[\vec u\times\vec\nabla E+\left(\vec\omega\times\vec u\right)\times\vec u\right]=\omega^i\Delta_i+u^i\epsilon^{ijk}D_j\Delta_k=0
\end{align}
when $\Delta=0$. The continuity equation generated by the sub-symmetry operator $X_h$ is the helicity conservation law.

\section{Conclusions}
We discussed the problem of correspondences between sub-symmetries and conservation laws for differential systems without well defined Lagrangian functions (quasi-Noether systems). We demonstrated that sub-symmetries are naturally related to conservation laws, and proved a theorem that associates to each sub-symmetry of a quasi-Noether system a local conservation law. We applied this approach to the Euler equations in both velocity and vorticity formulations. For Euler and constrained Euler systems, we obtained several classes of infinite sub-symmetries and generated corresponding classes of infinite conservation laws parametrized by arbitrary functions of velocities or vorticities. We showed that our set includes all previously known infinite conservation laws of the incompressible Euler equations;  we generated the infinite Casimir invariants of 2D flow, the conservation laws of generalized momenta/angular momenta, generalized helicity, and other classes of infinite conservation laws obtained in \cite{Kelbin}; (the infinite conservation laws obtained in \cite{Cheviakov} are discussed in \cite{Rosenhaus15}). We also showed that known helicity conservation law that could not be explained using a symmetry mechanism naturally arises from sub-symmetries.

\smallskip

We have found new classes of infinite conservation laws, \eqref{euler:con1}, and \eqref{vort:con}, and obtained a set of possible constraints on incompressible flows to possess infinite sub-symmetries and infinite conservation laws. The very fact that the essentially non-linear system of Euler equations admits infinite sub-symmetries \eqref{euler:subsym1} with arbitrary functions of dependent variables (and corresponding infinite conservation laws \eqref{euler:con1}) seems rather remarkable. Among other interesting consequences of this study is the existence of some spatial ``invariant" directions along which a flow exhibits an infinite sub-symmetry and an infinite set of conservation laws.

\section*{Acknowledgments}
One of the authors (V.R.) is grateful to M. Oberlack for helpful discussions in June of 2013.

\end{document}